\documentclass[conference]{IEEEtran}
\IEEEoverridecommandlockouts
\usepackage{cite}
\usepackage{amsmath,amssymb,amsfonts}
\usepackage{algorithmic}
\usepackage{graphicx}
\usepackage{textcomp}
\usepackage{xcolor}
\def\BibTeX{{\rm B\kern-.05em{\sc i\kern-.025em b}\kern-.08em
    T\kern-.1667em\lower.7ex\hbox{E}\kern-.125emX}}

\usepackage[shortlabels]{enumitem}
\usepackage{graphicx} 
\usepackage{multicol}
\usepackage{epsfig} 
\usepackage{amsmath,amsthm} 
\usepackage{amssymb}  
\usepackage{float}
\usepackage{cases,setspace,adjustbox}
\usepackage{cite}
\usepackage{algorithmic}
\usepackage[ruled,vlined]{algorithm2e}
\usepackage{array}
\usepackage[update,prepend]{epstopdf}
\usepackage{eqparbox}
\usepackage{fixltx2e}
\usepackage{url}
\usepackage{amsthm}
\usepackage{array}
\usepackage{caption}
\usepackage{xcolor}
\usepackage{mathtools}
\usepackage{bbm}
\usepackage{multirow}

\usepackage{booktabs}

\usepackage[font=normal]{subfig}

\usepackage{texdef2020}

\newcommand{\p}[1]{p_{#1}}
\newcommand{\pvec}{\vvec{p}}
\newcommand{\agevec}{\vvec{\Delta}_t}

\newcommand{\age}[2]{\Delta_{#1,#2}}
\newcommand{\minmaxpayoff}{v_{k,t}^*}
\newcommand{\lidle}{\sigma_\text{I}}
\newcommand{\lsucc}{\sigma_\text{S}}
\newcommand{\lcol}{\sigma_\text{C}}
\newcommand{\SlotExpectedAge}[1]{E[\age{k}{t}({#1})|\age{k}{0}]}
\newcommand{\discountfactor}{\alpha}
\newcommand{\ageconstant}{\delta_k}
\newcommand{\allstrategyp}{\mathcal{P}}
\newcommand{\AoI}[1]{\Delta_{#1}}		
\newcommand{\probabilityvector}{\mathbf{P}}

\theoremstyle{definition}
\newtheorem{lem}{Lemma}
\newtheorem{theorem}{Theorem}

\newtheorem{definition}{Definition}

\begin{document}
\title{Spectrum Sharing For Information Freshness:\\ A Repeated Games Perspective}
\author{Shreya Tyagi$^{*}$, Sneihil Gopal$^{\dagger\ddagger}$, Rakesh Chaturvedi$^{\star}$, Sanjit K. Kaul$^{*}$, \\
$^{*}$Wireless Systems Lab, IIIT-Delhi, India\\
$^{\dagger}$PREP Associate, Communications Technology Laboratory, National Institute of Standards and Technology, USA\\
$^{\ddagger}$Department of Physics, Georgetown University, USA\\
$^{\star}$Department of Social Sciences \& Humanities, IIIT-Delhi, India\\
shreyat@iiitd.ac.in, sneihil.gopal@nist.gov, \{rakesh, skkaul\}@iiitd.ac.in}
\maketitle
\begin{abstract}
We consider selfish sources that send updates to a monitor over a shared wireless access. The sources would like to minimize the age of their information at the monitor. Our goal is to devise strategies that incentivize such sources to use the shared spectrum cooperatively. Earlier work has modeled such a setting using a non-cooperative one-shot game, played over a single access slot, and has shown that under certain access settings the dominant strategy of each source is to transmit in any slot, resulting in packet collisions between the sources' transmissions and causing all of them to be decoded in error at the monitor. 

We capture the interaction of the sources over an infinitely many medium access slots using infinitely repeated games. We investigate strategies that enable cooperation resulting in an efficient use of the wireless access, while disincentivizing any source from unilaterally deviating from the strategy. Formally, we are interested in strategies that are a subgame perfect Nash equilibrium (SPNE). We begin by investigating the properties of the one-stage (slot) optimal and access-fair correlated strategies. We then consider their many-slot variants, the age-fair and access-fair strategies, in the infinitely repeated game model. We prove that the access-fair and age-fair strategies are SPNEs for when collision slots are longer than successful transmission slots. Otherwise, neither is a SPNE. We end with simulations that shed light on a possible SPNE for the latter case.
\end{abstract} 
\section{Introduction}
The growing demand for applications that require real-time monitoring and actuation necessitates investigation of how such applications may share the scarce wireless spectrum resource. For example, Industry 4.0~\cite{iiot} aims to achieve intelligent manufacturing processes by deploying Cyber-Physical Systems (CPS). Other examples include traffic monitoring and control, environmental monitoring, healthcare, smart homes, and networks of autonomous vehicles. We have sensors/devices communicating measurements to servers/aggregators or to other devices. The servers process the measurements and communicate actuation commands to actuators (agents/robots/plants) that execute the commands in their deployment environment. The servers require that sensor measurements available to them are as fresh as possible. Similarly, actuators require that actuation commands are as fresh as possible at their end.

In this work, we abstract out application-specific detail and focus on enabling freshness of information (sensed measurements/ actuation commands) at a \textit{monitor} (server/ actuator) when \textit{sources} share a wireless access network. Specifically, we consider a network of selfish sources, where each source sends its information \emph{updates} to a monitor over a shared CSMA/CA like wireless access. Each source requires that its updates available at the monitor are as fresh as possible.

We quantify freshness using the metric of age of information~\cite{yates2020age-survey}. Let $u_{k}(t)$ be the timestamp of the most recent information update of source $k$ at the monitor at time $t$. The age of updates of source $k$ at the monitor is a stochastic process $\Delta_{k,t} = t - u_{k}(t)$. Each source would like to optimize the sum of expected discounted age of its updates at the monitor, over a time horizon of interest.

Earlier work~\cite{gopal-kaul-roy-chaturvedi-INFOCOM-2020} has modeled the sharing of a wireless access by such sources. The interaction between the sources over any wireless medium access slot was modeled as a non-cooperative one-shot game~\cite{game_theory_cambridge_book}. However, under certain medium access settings, the dominant strategy for any node in any slot was to transmit its update. In shared access settings, the interfering transmissions result in a \textit{collision} slot that has all transmitted updates decoded in error at the monitor, resulting in wastage of the shared spectrum.

Selfish agents may be incentivized to \emph{cooperate} when they repeatedly interact over a long enough time-horizon~\cite{game_theory_cambridge_book}. In this paper, we investigate repeated interaction between sources sharing wireless access over an infinitely many medium access slots. We model the interaction as an infinitely repeated game. Our goal is to devise strategies that result in an efficient use of the wireless spectrum by the sources while disincentivizing any source to unilaterally deviate from the strategy. Formally, we are interested in strategies that are a subgame perfect Nash equilibrium (SPNE). 

We will separately consider two medium access settings that result from relative lengths of the collision slot and the successful transmission slot (exactly one source transmits). Medium access in which a collision slot is at least as long as a successful transmission slot is exemplified by the basic access scheme in the distributed control function (DCF)~\cite{bianchi} of IEEE 802.11 WLAN standard. The case when a collision slot is shorter is exemplified by the RTS/CTS scheme in DCF. Our specific contributions include:
\begin{enumerate}
    \item We define the one-stage game, parameterized by an age vector, that is played by sources sharing the wireless medium access in every slot. We propose the one-stage optimal and the access-fair correlated strategies. 
    \item Correlated strategies are often used to enable cooperation in repeated game settings. For each strategy, we determine the conditions (Lemma~\ref{lem:vec_pIsIndivRational}) that ensure the payoffs obtained are individually rational. We solve for the probability vector that corresponds to the one-stage optimal strategy (Theorems~\ref{thm:optimalPproperties},~\ref{thm:optimalPproperties_lsucc_leq_lcol},~\ref{thm:optimalPproperties_lsucc_gt_lcol}) for the two different access settings.
    \item We define the infinitely repeated cooperation game, for which we study the two strategies of age-fair and access-fair. We prove that both the strategies are SPNE for the setting when a successful transmission slot is shorter than a collision slot (Theorems~\ref{thm:access-fair-spne} and~\ref{thm:age-fair-spne}).
    \item We prove the strategies are \textbf{not} an SPNE (Theorem~\ref{thm:notSPNElSuccGtlCol}) for the alternate setting. We illustrate using simulations that the one-stage optimal strategy repeated every slot is an SPNE for a large range of discount factors, but only when a few sources share the network.
\end{enumerate}

The rest of the paper is organized as follows. Section~\ref{sec:related} summarizes related works. In Section~\ref{sec:model} we describe the network model. Section~\ref{sec:one-stage} details the one-stage game played by the sources in any medium access slot. We detail the proposed correlated cooperation strategies in Section~\ref{sec:correlated}. We define the infinite repeated game model and analyze and evaluate corresponding strategies in Section~\ref{sec:repeated-game}.
\section{Related Works}
\label{sec:related}
Recently there have been many works~\cite{yates2020age-survey} that contribute toward enabling freshness over networks, with freshness quantified using metrics of age of information at recipients. In~\cite{Kaul-Yates-isit2017,kosta2019age,chen2022age} authors investigated age for networks with multiple users sharing a slotted system. In~\cite{Kaul-Yates-isit2017} and~\cite{kosta2019age} authors considered scheduled and random access mechanisms. 
In~\cite{chen2022age} authors proposed distributed age-efficient transmission policies with the objective of minimizing the age over a random access channel. Similar to~\cite{Kaul-Yates-isit2017,kosta2019age,chen2022age}, we consider a CSMA/CA based access mechanism that sources use to compete for access to shared wireless spectrum.

In~\cite{nguyen-kompella-kam-jeffery-INFOCOM-2018,YinAoI2018,SGAoI2018,gopal-kaul-chaturvedi-roy-ACM-2021,kumar-vaze-IEEEjournal-2021,zheng-xiong-fan-zhong-globecom-2019,gopal-kaul-roy-chaturvedi-INFOCOM-2020} authors studied games with age as the payoff function. In~\cite{nguyen-kompella-kam-jeffery-INFOCOM-2018} and~\cite{YinAoI2018}, authors investigated adversarial settings where one player's objective is to maintain information freshness and the other player's aim is to prevent this. In~\cite{SGAoI2018}, we formulated a one-shot game to study the coexistence of Dedicated Short Range Communication (DSRC) and WiFi, where the DSRC network desires to minimize the average age of information and the WiFi network aims to maximize the average throughput. In~\cite{gopal-kaul-chaturvedi-roy-ACM-2021} using the repeated game model we provided insights, additional to~\cite{SGAoI2018}, on the interaction of age and throughput optimizing networks. In~\cite{kumar-vaze-IEEEjournal-2021} authors proposed a distributed transmission strategy for a collision model that minimizes age and transmission cost. In~\cite{zheng-xiong-fan-zhong-globecom-2019} authors considered a multi-helper-assisted wireless powered network and proposed a leader-follower strategy to establish efficient cooperation between the charging devices and the age-optimizing sensor-AP communication pair. In~\cite{gopal-kaul-roy-chaturvedi-INFOCOM-2020} we studied a one-shot multiple access game with multiple nodes and provided insights into how competing nodes that value timeliness would share the spectrum under different medium access settings. However, under certain medium access settings the strategy proposed in~\cite{gopal-kaul-roy-chaturvedi-INFOCOM-2020} led to spectrum wastage. In contrast to~\cite{gopal-kaul-roy-chaturvedi-INFOCOM-2020}, in this work we propose a repeated game to model the interaction between age optimizing nodes and devise strategies that are SPNE and result in an efficient use of wireless spectrum.
\section{Network Model}
\label{sec:model}
We have $n$ sources indexed $1,2,\ldots,n$ sharing a CSMA/CA like wireless slotted medium access to send their updates to a monitor. Sources attempt transmission at the beginning of a slot. We assume all sources can sense each other's transmissions. If exactly one source transmits during a slot, we assume that the source's update is correctly decoded by the monitor at the end of the slot. Such a slot is said to be a successful transmission slot. Let $S_k$ be the event that source $k$ transmitted successfully during a slot and let $S_{-k}$ be the event that exactly one source other than $k$ transmitted successfully during a slot. We assume that a source is able to generate-at-will a fresh update before transmission. Thus, the age of an update transmitted over a slot is $0$ at the beginning of the slot. 

A slot is an idle slot (event $I$) in case no sources transmit updates during the slot. Interference between sources, as they share the wireless access, is captured using a collision channel model. Specifically, in case more than one source transmits during a slot, none of the transmitted updates are decoded correctly by the monitor. The resulting slot is said to be a collision slot (event $C$).

The idle, successful transmission, and collision slots have lengths $\lidle$, $\lsucc$, and $\lcol$, respectively. The length $\lidle$ of an idle slot is typically much smaller than successful transmission and collision slots. Depending on the access mechanism, a successful transmission slot may be longer ($\lsucc > \lcol$) or shorter ($\lsucc \le \lcol$) than a successful transmission slot. In what follows, we will often consider the two cases separately.
\section{The One-Stage Game}
\label{sec:one-stage}

\begin{figure}[t]
\begin{center}
\includegraphics[width=0.9\columnwidth]{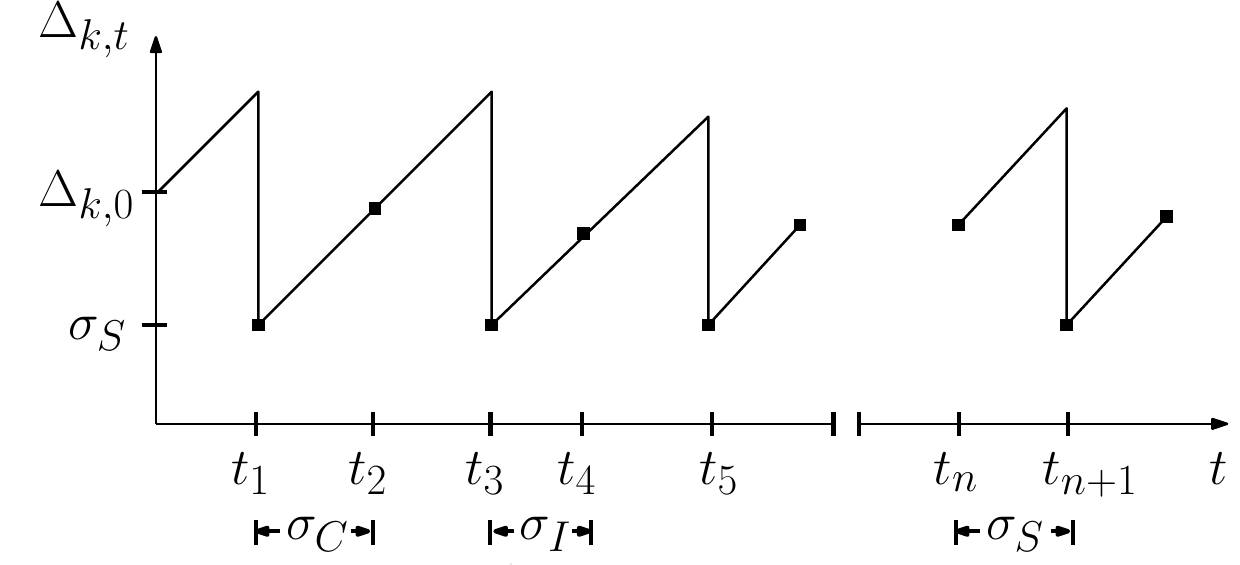}
\end{center}
\caption{\small Sample path of age $\AoI{k,t}$ of source $k$'s update at the monitor. $\AoI{k,0}$ is the initial age. When only source $k$ transmits in a slot, it sees a successful transmission, and its age resets to $\lsucc$. Otherwise, its age increases by $\lsucc$, $\lcol$, or $\lidle$ depending on whether the event $S_{-k}$, $C$, or $I$ took place. The time instants $t_{n}$, $n = 1,2,\ldots$, show the slot boundaries. A collision slot starts at $t_1$, an idle slot at $t_3$, and a slot in which the source $k$ transmits successfully starts at $t_n$.}
\label{fig:instaAoI}%
\vspace{-0.2in}
\end{figure}

The most basic interaction between sources in any slot is usefully modeled as a one-shot game. Let $\vvec{\Delta}_t = \vseq{\Delta}{1,t}{n,t}$ be the vector of ages of updates of sources $1,2,\ldots,n$ at the monitor at time $t$. Let $A_{k,t} = \{\mathcal{T}, \mathcal{I}\}$ be the set of pure strategies of any source $k$. If source $k$ chooses action $\mathcal{T}$, it transmits an update during the slot, else it idles during the slot.

The age $\age{k}{t+1}$ of source $k$'s updates at the monitor at the end of slot $t$ is a function of its age at beginning of the slot $t$ and the choices of actions by the $n$ sources during slot $t$. We have 
\begin{align}
\age{k}{t+1} = 
\begin{cases}
    \age{k}{t} + \lidle & I,\\
    \age{k}{t} + \lcol & C,\\
    \age{k}{t} + \lsucc & S_{-k},\\
    \lsucc & S_k.
\end{cases}
\label{eqn:ageEvolution}
\end{align}
The age increases at the monitor by the length of idle slot in case the event $I$ occurs, that is all sources choose the action $\mathcal{I}$ in slot $t$. In case more than one source chooses the action $\mathcal{T}$ in the slot, the slot becomes a collision slot (event $C$) of length $\lcol$. The age therefore increases by $\lcol$ at the end of such a slot. In case exactly one source other than source $k$ transmits in the slot, the event $S_{-k}$ takes place. The resulting slot is a successful transmission slot with length $\lsucc$. The age of source $k$ increases by $\lsucc$. The remaining case corresponds to when only source $k$ transmits (action $\mathcal{T}$) in the slot. The event $S_k$ occurs. As a result, the monitor receives an update from $k$ at the end of the slot that was fresh at the beginning of the transmission. The reception of the update resets the age of source $k$'s updates to the length $\lsucc$ of the successful transmission slot. Figure~\ref{fig:instaAoI} shows an example sample path of the age $\AoI{k,t}$.

The one-stage game played in slot $t$ is $G_t = (\mathcal{N}, A_t, \vvec{U}_t; \agevec)$. Here $\mathcal{N} =   \{1,2,\ldots, n\}$ is the set of sources (the players). The set $A_t = A_{1,t} \times A_{2,t}  \times \ldots  \times A_{n,t}$ is the profile of strategy spaces of the players. Each \emph{strategy profile} $a\in A_t$ is an $n$-tuple with element $k$ corresponding to an action in $A_k$. Associated with every $a\in A_t$ is a vector $\vvec{U}_t = \begin{bmatrix} U_{1,t}(a) & \ldots & U_{n,t}(a) \end{bmatrix}$ of payoffs, where $U_{k,t}(a)$ is the payoff obtained by source $k$. The $n$-tuple $a$ comprises of action $a_{k}$ by source $k$ and the $(n-1)$-tuple $a_{-k}$ of actions by the other sources, where $-k$ denotes the set of all sources $j\ne k$. The payoff obtained by source $k$ at the end of slot $t$ is determined by its choice of action $a_k$ and also the actions $a_{-k}$ chosen by the other sources.

The one-stage game $G_t$ is parameterized by the age vector $\agevec$, since this together with the actions chosen by the sources, determines the payoffs obtained by the sources at the end of slot $t$. The payoffs are $U_{k,t}(a) = - \age{k}{t+1}$, $k\in \mathcal{N}$, where $\age{k}{t+1}$ is obtained using Equation~(\ref{eqn:ageEvolution}).

We consider strategy profiles that are individually rational. That is $a \in A_t$ satisfies $U_{k,t}(a) \geq \minmaxpayoff$, for all sources $k$, where 
\begin{align}
    \minmaxpayoff = \min_{a_{-k} \in A_{-k,t}} \max_{a_k \in A_{k,t}} U_{k,t}(a_k,a_{-k})
    \label{eqn:minmaxpayoff}
\end{align}
is the minmax payoff of source $k$.

Minmax payoff $\minmaxpayoff$ is the minimum payoff that a payoff maximizing source $k$ obtains at the end of the slot, given age of its updates at the beginning of the slot. The minimum is calculated over the set containing the maximum payoff of $k$ for every $a_{-k} \in A_{-k,t}$. The payoffs $\minmaxpayoff$, $k\in \mathcal{N}$, are the basis for defining individual rationality as these are the worst possible payoffs that other sources can inflict on a source.

\begin{figure}[t]
\includegraphics[width = \columnwidth]{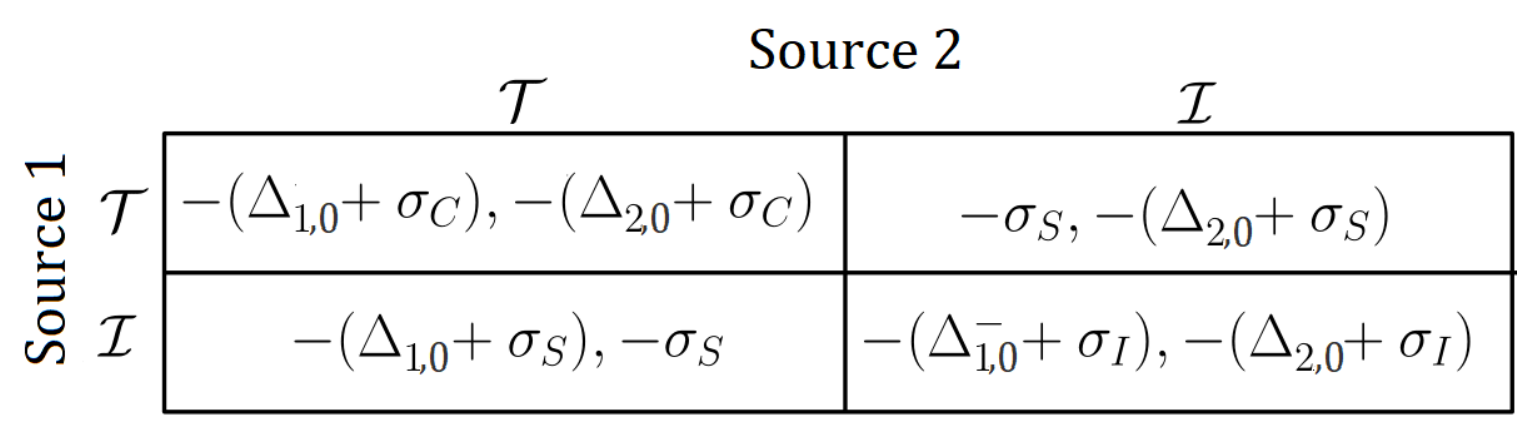}
\caption{\small Payoff matrix for the game $G$ when $\mathcal{N} = \{1,2\}$. The game is parameterized by the age vector $\begin{bmatrix}\Delta_{1,0} & \Delta_{2,0}\end{bmatrix}$ at the beginning of the slot. When both sources transmit, a collision results in their ages increasing by $\lcol$. The source payoffs at the end of the slot are $(- \Delta_{1,0} - \lcol, - \Delta_{2,0} - \lcol)$. The diagonal entries correspond to when exactly one of the sources transmits. The remaining entry is for when neither source transmits, which results in an idle slot.}
\label{fig:payoff_2player}
\end{figure}

Consider the example payoff matrix in Figure~\ref{fig:payoff_2player}. The minmax payoff for source $k$, $k=1,2$, is $-(\Delta_{k}^{-} + \lsucc)$ when $\lsucc \leq \lcol$ and $-(\Delta_{k}^{-} + \lcol)$ when $\lsucc > \lcol$. When we have $n > 2$ sources, the strategy profile $a_{-k}$ may be one in which none of the sources in $-k$ transmit during the slot, or one in which exactly one of the sources in the set $-k$ transmits, or one where more than one source in the set $-k$ transmits. For when $\lsucc \ge \lcol$, we must calculate the minimum over the set of payoffs $\{-\lsucc, -(\Delta_{k,t} + \lcol)\}$. When $\lsucc < \lcol$, the set of payoffs is $\{-\lsucc, -(\Delta_{k,t} + \lsucc), - (\Delta_{k,t} + \lcol)\}$. For both the cases, the minimum payoff is $-(\Delta_{k,t} + \lcol)$ and is independent of the relative ordering of the lengths of collision and successful transmission slots. Here we assumed, as we will throughout the paper, that the age at the beginning of a typical slot is $\ge \lsucc$. This is because age can't be reset unless a successful transmission takes places and it is reset to $\lsucc$.

In what follows, we will take the feasible payoff set to be $\mathcal{F} =\{\vvec{U}_t: \forall k \in \mathcal{N}, U_{k,t} \geq \minmaxpayoff\}$, which is the set of individually rational payoff vectors. This is motivated by the Folk Theorem~\cite{Friedman-james-TheReviewofEconomicStudies-1971}, which states that, in a repeated game setting, any payoff vector within the feasible payoff set can be supported as a subgame perfect Nash equilibrium (defined later) provided the players are sufficiently patient. Essentially, such payoff vectors can incentivize the selfish sources to cooperatively share the shared spectrum, for an appropriate discounting of future payoffs.

In the following section, we define two \emph{correlated} strategy profiles for a stage game. Later in Section~\ref{sec:repeated-game} we analyze the strategy profiles to determine whether they constitute a SPNE when played repeatedly over an infinite medium access slots. 
\section{Correlated Strategies}
\label{sec:correlated}
The use of correlated strategies, which are interpreted as strategies that players use based on commonly observed but unspecified signals, is often used in supporting cooperation in repeated games.
\begin{definition}
\textbf{Correlated Strategy} A probability distribution over the set $A_t$ of strategy profiles for the one-stage game.
\end{definition}
\noindent Define probabilities $p_k = P[S_k]$, $k = 1, \ldots, n$. Let $\p{I}$ and $\p{C}$ be the probabilities of an idle slot and a collision slot. Let $\vvec{p} = \vec{\p{1} & \ldots & \p{n} & \p{I} & \p{C}}$ be a probability vector such that $\sum_{k=1}^n \p{k} + \p{I} + \p{C} = 1$. Define the payoff $U_{k,t}(\vvec{p})$ obtained by source $k$ when strategy profiles are randomized using $\vvec{p}$.
\begin{align}
    U_{k,t}(\vvec{p}) = E_{a \sim \vvec{p}}[U_{k,t}(a)|\Delta_{k,t}].
\end{align}
Note that the way $\vvec{p}$ is defined, we don't distinguish between strategy profiles $a\in A_t$ that result in the collision event $C$. All such strategy profiles result in an increase by $\lcol$ of all sources' ages at the monitor. Let $\vvec{U}_t(\vvec{p}) = \vec{U_{1,t}(\vvec{p}) & \ldots & U_{n,t}(\vvec{p})}$ be the payoff vector. We must consider distributions from the set $\allstrategyp = \{\vvec{p}: \vvec{U}_t(\vvec{p}) \in \mathcal{F}\}$.

We consider two correlated strategies that we refer to as (a) one-stage optimal $O$ and (b) access-fair $U$. When following the one-stage optimal strategy the sources choose a strategy profile $a\in A_t$ that is drawn from the distribution defined by the probability vector $\vvec{p}^*$ that solves the problem
\begin{align}
   \underset{\vvec{p} \in \allstrategyp}{\text{Maximize}} \sum_{k=1}^{n} U_{k,t}(\vvec{p}).
   \label{opt:onestageopt}
\end{align}

The access-fair strategy uses $\vvec{p} = \begin{bmatrix} 1/n & \ldots & 1/n & 0 & 0 \end{bmatrix}$. That is any source is chosen to transmit in the slot with probability $1/n$. While simple to implement, as it doesn't need the ages at the beginning of the slot, $\vvec{p}$ may not lie in the set $\allstrategyp$ and thus may not result in individually rational payoffs. Lemma~\ref{lem:vec_pIsIndivRational} provides useful conditions on $\vvec{p}$.

\begin{lem}
Conditions on $\vvec{p}$ that ensure individually rational payoffs are as follows.
\begin{enumerate}
    \item $\lsucc \leq \lcol$. For $n=2$ sources, $\allstrategyp = \{\pvec: \p{C} = 0\}$ is a sufficient condition for $\pvec$ to be in $\allstrategyp$. When $n> 2$, any probability vector $\vvec{p}$ is in $\allstrategyp$.
    \item $\lsucc > \lcol$. The access-fair strategy is in $\allstrategyp$ if and only if $\age{k}{t} \ge n(\lsucc - \lcol)$.
\end{enumerate}
\label{lem:vec_pIsIndivRational}
\end{lem}


Next we present Theorems~\ref{thm:optimalPproperties},~\ref{thm:optimalPproperties_lsucc_leq_lcol}, and~\ref{thm:optimalPproperties_lsucc_gt_lcol} that specify the properties of $\vvec{p}^*$.

\begin{theorem}
    For any ordering of $\lsucc$ and $\lcol$, the probability vector $\vvec{p}^*$, which solves Problem~(\ref{opt:onestageopt}) and is the one-stage optimal strategy, has $\p{C} = 0$. In addition, if for every source $k$, $\age{k}{t} < n(\lsucc - \lidle)\,$, $\p{I} = 1$.
    \label{thm:optimalPproperties}
\end{theorem}

\begin{theorem}
    When $\lsucc \leq \lcol$, $\vvec{p}^*$ chooses the source with the maximum age at the beginning of the slot to transmit an update with probability $1$.
    \label{thm:optimalPproperties_lsucc_leq_lcol}
\end{theorem}

\begin{theorem}
    When $\lsucc > \lcol$:
    \begin{enumerate}
        \item If the harmonic mean of source ages, $N\prod_{k=1}^{N} \age{k}{t}/\sum_{k=1}^{N} (\prod_{j=1, j \neq k}^{N} \age{j}{t})$, is greater than $N(\lsucc - \lidle)$, then probability $p_k$ for all sources, except the one with the maximum age at the beginning of the slot, is given by $p_k = \frac{\lsucc - \lcol}{\age{k}{t}}$. This probability for the source with the maximum age is $1 - \sum_{i=1, i \neq j}^{N}p_i$. Note that this implies that $p_I = 0$.
        
        \item If the harmonic mean is less than $N(\lsucc - \lidle)$ then $p_I = \frac{\lsucc -\lcol}{\lsucc -\lidle}$ and the source with the maximum age transmits with probability $1 - p_I$. 
    \end{enumerate}
    \label{thm:optimalPproperties_lsucc_gt_lcol}
\end{theorem}
The proofs of Lemma~\ref{lem:vec_pIsIndivRational} and Theorems~\ref{thm:optimalPproperties},~\ref{thm:optimalPproperties_lsucc_leq_lcol}, and~\ref{thm:optimalPproperties_lsucc_gt_lcol} can be found in the appendix.
\section{Infinitely Repeated Cooperation Game}
\label{sec:repeated-game}
In an infinitely repeated game $G^\infty$, the stage game $G$ is repeated over infinitely many slots indexed $t = 0, 1, \ldots$. As has been observed, for example in~\cite{game_theory_cambridge_book}, a repeated game model best represents the interaction of multiple sources sharing a network. The repeated game begins with slot $t = 0$ with age $\vvec{\Delta}_0$ at the beginning of the slot. Sources choose actions from their time-invariant strategy sets $\{\mathcal{T}, \mathcal{I}\}$. Consequently, the age vector gets updated to $\vvec{\Delta}_1$ at the end of slot $0$, which becomes an input to the stage game in slot $t=1$, and so on. 

The age vector $\vvec{\Delta}_t$ serves as the state of the model that fully summarizes the payoff-relevant interactions of the sources up to the beginning of slot $t$. Consequently, in our infinitely repeated game, each slot is identical except for the age vector at the beginning of the slot. A sample path of play in the repeated game is an infinite sequence $s = (s_0, \ldots, s_t, \ldots)$, where $s_t \in A_t$, $A_t$ is the profile of strategy spaces defined for the one-shot game $G$ in Section~\ref{sec:one-stage}. Source $k$'s utility in the repeated game setting is the discounted sum payoff
\begin{align}
    U_{k,\infty}(s|\age{k}{0}) = -(1 - \discountfactor)\sum_{t=1}^\infty \discountfactor^{t-1} \age{k}{t}(s_{t-1}, \age{k}{t-1}),
\end{align}
for sample path of play $s$ and discount factor $0 \le \alpha < 1$.

A strategy in the repeated game is a complete contingent plan of action -- a prescription of what to do given any slot and state. A subgame perfect Nash equilibrium of $G^\infty$ is a profile of strategies in the repeated game such that given any slot and the state (age vector) at its beginning, the continuation strategies from that slot onwards constitute a Nash equilibrium in the continuation game.

We study two strategies, the access-fair $U^\infty$ and age-fair $M^\infty$, which are quite simple from the perspective of repeated games in that their prescription is invariant to slot and state. In the access-fair strategy $U^\infty$, in any slot and for any age vector (state) at the beginning of a slot, any source is selected uniformly and randomly with probability $1/n$ to transmit its update. Essentially, we repeat the one-stage correlated strategy $U$ for all slots. As a result, in any slot the chosen source transmits while others idle. In the age-fair strategy $M^\infty$, in any slot and for any age vector at the beginning of a slot, the source whose age of updates at the monitor is maximum, is chosen to transmit with probability $1$. Note that, from Theorem~\ref{thm:optimalPproperties_lsucc_leq_lcol}, when $\lsucc \le \lcol$, strategy $M^\infty = O^\infty$, where $O$ was defined as the one-stage optimal strategy obtained by solving~(\ref{opt:onestageopt}). However, from Theorem~\ref{thm:optimalPproperties_lsucc_gt_lcol}, $O$ isn't the same as $M$ when $\lsucc > \lcol$.

The repeated game model has discounting and bounded payoffs per slot. Therefore, by the one-shot deviation principle, which is simply a game-theoretic version of the Bellman's principle of optimality used in dynamic programming, the following conclusion holds. Let $C^\infty = (C, C, \ldots)$ be a repeated game path of play generated by playing strategy $C$ repeatedly. Let $D_kC^\infty = (D_kC_{-k}, C, \ldots)$ be the repeated game path of play generated when source $k$ deviates in the current slot while other sources continue to follow $C$ and subsequently, from the next slot onward, all sources follow $C^\infty$. The repeated game strategy $C^\infty$ is an SPNE if, for every source $k$, $U_{k,\infty}(C^\infty|\Delta_{k,t}) \ge U_{k,\infty}(D_kC^\infty|\Delta_{k,t})$.

\subsection{Discounted Payoff for Strategy $U^\infty$}
At the end of any slot, age for source $k$ resets to $\lsucc$ with probability $1/N$. The conditional PMF of age of source $k$ in slot $t$, given $\age{k}{0}$ is
\begin{align*}
\small
P[\age{k}{t}(U^\infty) = x| \age{k}{0} = \ageconstant] = 
\begin{cases} 
\frac{(N-1)^{m-1}}{N^m},&x = m\lsucc, \\
& m = 1,\ldots,t,\\
\frac{(N-1)^t}{N^t},&x= \ageconstant + t\lsucc.
\end{cases}
\end{align*}

The first term in the distribution represents a successful transmission for source $k$ followed by $m - 1$ successful transmissions for sources other than $k$. The last term captures the case where source $k$ is not selected to transmit at all, increasing the age by $t\lsucc$. The expected age $\SlotExpectedAge{U^\infty}$ of a source $k$ in any given slot $t$ can be calculated in closed-form (we skip the expression to optimize use of space).
The discounted payoff for source $k$ is
{
\begin{equation} \label{eqn:U_infinity_payoff}
    U_{k,\infty}(U^\infty|\age{k}{0}) = -(1-\discountfactor) \sum_{t=1}^{\infty} \discountfactor^{t-1} \SlotExpectedAge{U^\infty}. 
\end{equation}}

\subsection{Access and age-fair strategies are SPNE when $\lsucc \leq \lcol$}
Theorems~\ref{thm:access-fair-spne} and~\ref{thm:age-fair-spne} summarize our findings for when $\lsucc \leq \lcol$. Both the access-fair and the age-fair strategies are SPNE. They result in a good use of the shared access while disincentivizing a source to unilaterally deviate from the strategy.

\begin{theorem}
When $\lsucc \leq \lcol$, for any $\alpha \in [0,1)$ and any $n\in\{2,3,\ldots\}$, the access-fair strategy $U^\infty$ is an SPNE.
 \label{thm:access-fair-spne}
\end{theorem}

 \begin{proof}
A one-shot deviation by source $k$ can happen in the following ways. Source $k$ was selected to transmit in a slot but stayed idle. In this case, the age of source $k$ increments by $\lidle$. The other possible deviation is when a source other than $k$ is selected to transmit in a slot and $k$ also transmits during the slot. In this case, the age of source $k$ will increase by $\lcol$. Given the one-shot deviation principle, it suffices to assume that source $k$ deviated in slot $0$.

For the first type of deviation, the repeated game payoff obtained on deviating, is  $U_{k,\infty}(D_kU^\infty|\age{k}{0}) = -(\age{k}{0} + \lidle) + \discountfactor U_{k,\infty}(U^\infty|\age{k}{0} + \lidle)$, where $-(\age{k}{0} + \lidle)$ is the payoff obtained by the deviating source at the end of slot $0$. The second term is the repeated game payoff, discounted by $\alpha$, obtained slot $1$ onward for an age of $\age{k}{0} + \lidle$ at the beginning of slot $1$. In the absence of the deviation, when the source transmits instead of staying idle, the payoff is $U_{k,\infty}(U^\infty|\age{k}{0}) = -\lsucc + \discountfactor U_{k,\infty}(U^\infty|\lsucc)$. Note that $-\lsucc > -(\age{k}{0} + \lidle)$. This is because $\age{k}{0} \ge \lsucc$ and $\lidle > 0$. Also, from~(\ref{eqn:U_infinity_payoff}), $U_{k,\infty}(U^\infty|\lsucc) \ge U_{k,\infty}(U^\infty|\age{k}{0} + \lidle)$. Thus the deviation is not beneficial for any $0 \le \alpha < 1$.

The proof is similar for the second type of deviation.
 \end{proof}

\begin{theorem}
When $\lsucc \leq \lcol$, for any $\alpha\in [0,1)$ and any $n\in\{2,3,\ldots\}$, the age-fair strategy $M^\infty$ is an SPNE.
\label{thm:age-fair-spne}
\end{theorem}

\begin{proof}
The strategy $M^\infty$ selects the source with the highest age to transmit. The two possible one-shot deviations are (a) source $k$ was selected to transmit but stayed idle, leading to an age increment of $\lidle$ for all sources, and (b) some other source was selected for transmission but source $k$ also transmitted, leading to an age increment of $\lcol$ for all sources. 

Note that since both deviations result in an increase in the ages of all sources, the ordering of sources based on their ages stays unchanged. That is sources must take turns in the same order post the slot in which the deviation occurred, as they would have in the absence of a deviation.

Consider the first type of deviation. Since it results in an idle slot, the payoff of source $k$ when it deviates is $U_{k,\infty}(D_kM^\infty|\age{k}{0}) = -(\age{k}{0} + \lidle) + \discountfactor U_{k,\infty}(M^\infty|\age{k}{0} + \lidle)$. Given that the source $k$ had been chosen for transmission in slot $0$, the source had the largest age at the beginning of slot $0$ and will have the largest age at the beginning of slot $1$. It will be chosen to transmit in slot $1$, resulting in its age being reset to $\lsucc$. We can rewrite the above payoff as  $U_{k,\infty}(D_kM^\infty|\age{k}{0}) = -(\age{k}{0} + \lidle) - \discountfactor \lsucc +  \discountfactor^2 U_{k,\infty}(M^\infty|\lsucc)$. Compare this with the payoff in the absence of deviation, which is $U_{k,\infty}(M^\infty|\age{k}{0}) = -\lsucc +  \discountfactor U_{k,\infty}(M^\infty|\lsucc)$.  Further note that $\age{k}{0} \ge \lsucc$. Deviation results in a worse payoff for all $0 \le \alpha < 1$.

The second type of deviation results in a collision slot for all. The payoff that results from deviation is $U_{k,\infty}(D_kM^\infty|\age{k}{0}) = -(\age{k}{0} + \lcol) + \discountfactor U_{k,\infty}(M^\infty|\age{k}{0} + \lcol)$. This corresponds to a sequence of one-stage payoffs $-(\age{k}{0} + \lcol)$, $-(\age{k}{0} + \lcol + \lsucc)$, and so on, till the slot when source $k$ becomes the source with the largest age. At the end of this slot, the source will receive a payoff of $-\lsucc$. In the absence of deviation, the source gets a payoff $U_{k,\infty}(M^\infty|\age{k}{0}) = -(\age{k}{0} + \lsucc) + \discountfactor U_{k,\infty}(M^\infty|\age{k}{0} + \lsucc)$. The corresponding sequence of one-stage payoffs is $-(\age{k}{0} + \lsucc)$, $-(\age{k}{0} + \lsucc + \lsucc)$, and so on, till source $k$ gets its chance to transmit. Comparing the sequences of one-stage payoff for when the source deviates and when it doesn't, the source obtains larger one-stage payoffs ($\lsucc \le \lcol$) when it doesn't deviate. Eventually, both when the source deviates and otherwise, the source will transmit and its age will be reset to $\lsucc$, giving it a one-stage payoff of $-\lsucc$. However, it will take one additional slot to happen when the source deviates. Thus deviation makes the source's repeated game payoff worse.
\end{proof}

\subsection{SPNE for when $\lsucc > \lcol$ is elusive}

\begin{theorem}
When $\lsucc > \lcol$, neither the access-fair strategy $U^\infty$ nor the age-fair strategy $M^\infty$ is an SPNE.
\label{thm:notSPNElSuccGtlCol}
\end{theorem}

\begin{proof}
Consider $U^\infty$. It applies the strategy $U$ in every slot. Recall from Lemma~\ref{lem:vec_pIsIndivRational} that when $\lsucc > \lcol$, the payoff vectors from $U$ are not individually rational for all age vectors at the beginning of a slot. As a result, $U^\infty$ isn't an SPNE when $\lsucc > \lcol$. Next, consider $M^\infty$. This applies $M$ in every slot. Let source $i$ have the largest age at the beginning of a certain slot $t$. $M$ will choose this source for transmitting its update with probability $1$. For any other source $k\ne i$, the payoff at the end of the slot will be $-(\age{k}{t} + \lsucc)$, which is smaller than the minmax payoff of $-(\age{k}{t} + \lcol)$ when $\lsucc > \lcol$. The payoff is not individually rational. $M^\infty$ can't be a SPNE.
\end{proof}

Next we consider the strategy $O^\infty$ for when $\lsucc > \lcol$. $O^\infty$ has every stage game use the strategy $O$. Recall that the strategy $O$ is one-stage optimal and uses a probability vector in every slot that is obtained by solving~(\ref{opt:onestageopt}), described in Theorem~\ref{thm:optimalPproperties_lsucc_gt_lcol}. As per the Folk theorem, $O^\infty$ is an SPNE for \emph{sufficiently patient} players, since its payoff vectors lie in the set of individually rational payoffs. We are unable to analytically find the range of the discount factor $\discountfactor$ for which $O^\infty$ is a SPNE. However, Monte Carlo simulations show that $O^\infty$ isn't an SPNE for $\alpha \le 0.99$, even for a small number of sources. 
\begin{figure}[t]
  \centering
  {\includegraphics[width=0.9\columnwidth]{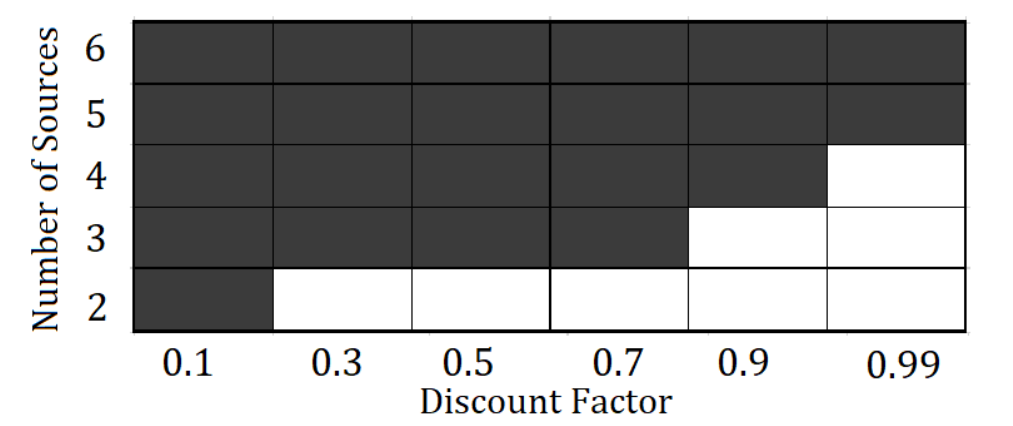}}\label{fig:heatmap}
\caption{\small The white cells correspond to $(n, \discountfactor)$ for which $O^\infty$ is an SPNE. We set $\lidle = \beta << 1$, $\lsucc = (1+\beta)$ and $\lcol = 0.1(1+\beta)$, with $\beta = 0.01$. The selection of slot lengths is such that the ratio $(\lsucc-\lidle)/\lsucc$ for the simulation setup is approximately the same as that for $802.11$ac based  devices and $802.11$p based vehicular networks.}
\label{fig:PayoffComparison}
\vspace{-0.15in}
\end{figure}

Figure~\ref{fig:PayoffComparison} depicts the region in which $O^\infty$ is an SPNE for number of sources $n\in \{2,3,4,5,6\}$ and $\discountfactor \in \{0.1, 0.3, 0.5, 0.7, 0.9, 0.99\}$. For each $(n,\alpha)$, we simulate $100$ initial age vectors, where in the elements of each vector are chosen uniformly and randomly in the range $[n, 3n]$. For each choice of initial age vector, we compare $U_{k,\infty}(O^\infty|\age{k}{0})$ and $U_{k,\infty}(D_kO^\infty|\age{k}{0})$. We choose the source $k=1$. Note that since we are only considering unilateral deviations by a certain source, it is sufficient to have source $1$ deviate and check whether the same results in better payoff than cooperation. The two infinite game payoffs, one when source $1$ cooperates and the other when source $1$ deviates, for a given initial age vector $\age{k}{0}$, are estimated to be the average of the sum-discounted payoffs that are obtained over $10000$ sample paths, each $10000$ slots long. The estimates were similar for when we simulated up to $50000$ slots in a sample path. 
\section{Conclusion}
We modeled the interaction between selfish sources sharing a wireless access as an infinitely repeated game. The sources would like the age of their information at a monitor to be as small as possible. We devise strategies for the repeated game setting that are proven to be a subgame perfect Nash equilibrium for when the successful transmission slot of the medium access is shorter than the collision slot. Such strategies enable cooperative use of the wireless spectrum while ensuring that no source has an incentive for one-shot deviation from the strategy. However, the SPNE strategies for the above setting, were shown not to be a SPNE for when the successful transmission slot is longer than the collision slot. This setting requires further investigation. 
\begin{spacing}{0.99}
\bibliographystyle{IEEEtran}
\bibliography{citations}
\end{spacing}
\appendix
\subsection{Proof of Lemma~\ref{lem:vec_pIsIndivRational}}
\label{FirstAppendix}
\begin{enumerate}
    \item Consider when $n=2$ and $\lsucc \le \lcol$.The minmax payoff~(\ref{eqn:minmaxpayoff}) for any node $k$, is $-(\age{k}{t} + \lsucc)$. Consider any probability vector $\vvec{p}$ with $p_C = 0$. The expected payoff that results from such a vector is $-(\lsucc p_k + (\age{k}{t} + \lsucc) p_{-k} + (\age{k}{t} + \lidle )\p{I})$. Since $\age{k}{t} + \lsucc$ is at least as large as each of $\lsucc, \age{k}{t} + \lsucc, \age{k}{t} + \lidle$, and $p_k + p_{-k} + \p{I} = 1$, $p_C = 0$ is sufficient for $\vvec{p}$ to result in individually rational payoffs.

    For when $n>2$ and $\lsucc \le \lcol$, the minmax payoff is $-(\age{k}{t} + \lcol)$. Any probability vector $\vvec{p}$ with probability $\p{C} \le 1$ will result in a payoff at least as large as the minmax payoff. Hence satisfying the individual rationality constraints. 
    
    \item Now consider $\lsucc > \lcol$. The minmax payoff is $-(\age{k}{t} + \lcol)$ for $n\ge 2$. The one-stage expected payoff on playing the uniform strategy is not individually rational if $- E[\age{k}{t+1}|\age{k}{t}] = - (1/n)\lsucc - ((n-1)/n)(\age{k}{t-1} + \lsucc) > - (\age{k}{t-1} + \lcol)$. That is if $\age{k}{t} < n(\lsucc - \lcol)$.
\end{enumerate}




\subsection{Proof of Theorem~\ref{thm:optimalPproperties}}
\label{Theorem 1}
\begin{enumerate}
    \item Suppose some probability vector $\probabilityvector^* = [\p{1}, ..., \p{N}, \p{I}, \p{C}]$ satisfies the constraints and is optimal with $\p{C} > 0 $, then changing $\p{C} = 0$ and $\p{I} = \p{C} + \p{I}$, i.e., another vector $[\p{1}, ..., \p{N}, \p{C} + \p{I}, 0]$ is in the feasible set and contradicts by achieving lower objective function value than the optimal vector $\probabilityvector^*$. Therefore an optimal probability vector will always have $\p{C} = 0$ 
    \item On expanding the objective function \ref{opt:onestageopt} we get, $\sum_{j=1}^{n}\age{j}{0} - \age{k}{0} + n\lsucc$ term corresponding to $p_k$, $\sum_{j=1}^{n}\age{j}{0} + n\lidle$ term corresponding to $p_I$. The coefficient corresponding to $\p{I}$ in the objective function is minimum if for all nodes $k$, $\Delta_k < N(\lsucc - \lidle)$. This implies setting $p_I = 1$ would minimize the objective function.
\end{enumerate}

\subsection{Proof of Theorem~\ref{thm:optimalPproperties_lsucc_leq_lcol}}
\label{Theorem 2}
Minimizing a linear function assigns maximum weight to the term with the minimum coefficient. Therefore, the optimal solution is where a node with the highest age gets to reset its age to $\lsucc$ with probability $1$. This lies in the feasible set (\ref{lem:vec_pIsIndivRational}).  

\subsection{Proof of Theorem~\ref{thm:optimalPproperties_lsucc_gt_lcol}}
\label{Theorem 3}
Consider the probability simplex of payoffs corresponding to pure strategies. The optimal solution  $\probabilityvector^*$ to objective function (\ref{opt:onestageopt}) lies on the boundary of this simplex. There can be two scenarios depending on the initial ages of nodes in the network.

Case 1: The payoff corresponding to a strategy where all nodes stay idle lies inside the simplex.

The equation of the hyperplane joining the payoffs corresponding to the strategies where one node successfully transmits is given by $H(\mathbf{n},b)$, where $\mathbf{n} = (\Delta_2\Delta_3...\Delta_N, \Delta_1\Delta_3...\Delta_N, \Delta_1\Delta_2...\Delta_{N-1})$ is normal to the hyperplane, and $b = \mathbf{n} \cdot \mathbf{x}$, where $\mathbf{x}$ is some point on the hyperplane. The point corresponding to the payoff where all nodes staying idle lie inside the simplex if the harmonic mean of source ages, $N\prod_{k=1}^{N} \age{k}{t}/\sum_{k=1}^{N} (\prod_{j=1, j \neq k}^{N} \age{j}{t}) \geq N(\lsucc - \lidle)$. Hence $p_I = 0$.

 Since $p_I = p_C = 0$, the individual rationality constraints reduce to $p_k \geq \frac{\lsucc - \lcol}{\age{k}{t}}$. In order to minimize the objective function (~\ref{opt:onestageopt}), We assign the lowest possible probability to all the nodes except the one with the maximum age. The node with the maximum age, denoted as $j$, the probability is set as $p_j = 1 - \sum_{i=1, i \neq j}^{N}p_i$. 
 
Case 2: Similarly, the payoff corresponding to a strategy where all nodes stay idle lies on the boundary of the simplex if $N\prod_{k=1}^{N} \age{k}{t}/\sum_{k=1}^{N} (\prod_{j=1, j \neq k}^{N} \age{j}{t}) < N(\lsucc - \lidle)$. The optimal solution would be on the boundary of the simplex, i.e., a line joining payoff corresponding strategy where all nodes sit idle and successful transmission for the node with a maximum age. Individually rational constraints result in $p_I \geq \frac{\lsucc -\lcol}{\lsucc -\lidle}$. In order to minimize the objective  function, We assign the lowest possible probability of all nodes idling to the lower bound i.e.,  $p_I = \frac{\lsucc -\lcol}{\lsucc -\lidle}$ and node with maximum age transmits with probability $1-p_I$.
\end{document}